\documentclass[10pt]{article}
\usepackage[a5paper, margin=.7in]{geometry}
\usepackage[T1]{fontenc}

\usepackage{amsmath,amssymb}
\usepackage{amsthm}

\usepackage{stmaryrd}

\usepackage{graphicx,xcolor}
\usepackage{hyperref}

\usepackage{algorithm,algpseudocode,float} 

\newtheorem{proposition}{Proposition}

\usepackage{authblk} 
\author{Mendes Oulamara\thanks{\url{mendes@deuxfleurs.fr}}~}
\author{Alex Auvolat\thanks{\url{alex@adnab.me}}}
\affil{Deuxfleurs}
\date{}

\title{An algorithm for geo-distributed and redundant storage in Garage}

\begin{document}

\maketitle


\begin{abstract}
	This paper presents an optimal algorithm to compute the assignment of data to storage nodes in the Garage geo-distributed storage system.
	We discuss the complexity of the different steps of the algorithm and metrics that can be displayed to the user.
\end{abstract}

\maketitle

\section{Introduction}

Garage\footnote{\url{https://garagehq.deuxfleurs.fr/}} is an open-source distributed object storage service tailored for self-hosting. It was designed by the Deuxfleurs association\footnote{\url{https://deuxfleurs.fr/}} to enable small structures (associations, collectives, small companies) to share storage resources to reliably self-host their data, possibly with old and non-reliable machines.
To achieve these reliability and availability goals, the data is broken into \emph{partitions} and every partition is replicated over 3 different machines (that we call \emph{nodes}). When the data is queried, it is fetched from one of the nodes. A \emph{replication factor} of 3 ensures good guarantees regarding node failure\cite{Raynal2018}. But this parameter can be another (preferably larger and odd) number.

Moreover, if the nodes are spread over different \emph{zones} (different houses, offices, cities\dots), we can require the data to be replicated over nodes belonging to different zones. This improves the storage robustness against zone failures (such as power outages). To do so, we define a \emph{scattering factor}, that is no more than the replication factor, and we require that the replicas of any partition are spread over this number of zones at least.

In this work, we propose an assignment algorithm that, given the nodes specifications and the replication and scattering factors, computes an optimal assignment of partitions to nodes. We say that the assignment is optimal in the sense that it maximizes the size of the partitions, and hence the effective storage capacity of the system.

Moreover, when a former assignment exists, which is not optimal anymore due to node or zone changes, our algorithm computes a new optimal assignment that minimizes the amount of data to be transferred during the assignment update (the \emph{transfer load}).

We call the set of nodes cooperating to store the data a \emph{cluster}, and a description of the nodes, zones and the assignment of partitions to nodes a \emph{cluster layout}

\subsection{Notations}

Let $k$ be some fixed parameter value, typically 8, that we call the ``partition bits''.
Every object to be stored in the system is split into data blocks of fixed size. We compute a hash $h(\mathbf{b})$ of every such block $\mathbf{b}$, and we define the $k$ first bits of this hash to be the partition number $p(\mathbf{b})$ of the block. This label can take $P=2^k$ different values, and hence there are $P$ different partitions. We denote $\mathbf{P}$ the set of partition labels (i.e. $\mathbf{P}=\llbracket1,P\rrbracket$).

We are given a set $\mathbf{N}$ of $N$ nodes  and a set  $\mathbf{Z}$ of $Z$ zones. Every node $n$ has a non-negative storage capacity $c_n\ge 0$ and belongs to a zone $z_n\in \mathbf{Z}$. We are also given a replication factor $\rho_\mathbf{N}$ and a scattering factor $\rho_\mathbf{Z}$ such that  $1\le \rho_\mathbf{Z} \le \rho_\mathbf{N}$ (typical values would be $\rho_N=\rho_Z=3$).

Our goal is to compute an assignment $\alpha = (\alpha_p^1, \ldots, \alpha_p^{\rho_\mathbf{N}})_{p\in \mathbf{P}}$ such that every partition $p$ is associated to $\rho_\mathbf{N}$ distinct nodes $\alpha_p^1, \ldots, \alpha_p^{\rho_\mathbf{N}} \in \mathbf{N}$ and these nodes belong to at least $\rho_\mathbf{Z}$ distinct zones. Among the possible assignments, we choose one that \emph{maximizes} the effective storage capacity of the cluster. If the layout contained a previous assignment $\alpha'$, we \emph{minimize} the amount of data to transfer during the layout update by making $\alpha$ as close as possible to $\alpha'$. These maximization and minimization are described more formally in the following section.

\subsection{Optimization objectives}

To link the effective storage capacity of the cluster to partition assignment, we make the following assumption:

\begin{equation}
	\tag{H1}
	\text{\emph{All partitions have the same size $s$.}}
\end{equation}

This assumption is justified by the dispersion of the hashing function, when the number of partitions is small relative to the number of stored blocks.

Every node $n$ will store some number $p_n$ of partitions (it is the number of partitions $p$ such that $n$ appears in the $\alpha_p$). Hence the partitions stored by $n$ (and hence all partitions by our assumption) have their size bounded by $c_n/p_n$. This remark leads us to define the optimal size that we will want to maximize:

\begin{equation}
	\label{eq:optimal}
	\tag{OPT}
s^* = \min_{n \in N} \frac{c_n}{p_n}.
\end{equation}

When the capacities of the nodes are updated (this includes adding or removing a node), we want to update the assignment as well. However, transferring the data between nodes has a cost and we would like to limit the number of changes in the assignment. We make the following assumption:

\begin{equation}
	\tag{H2}
	\text{\emph{Node changes happen rarely relatively to data block reads and writes.}}
\end{equation}

This assumption justifies that when we compute the new assignment $\alpha$, it is worth to optimize the partition size \eqref{eq:optimal} first, and then, among the possible optimal solutions, to try to minimize the number of partition transfers. More formally, we minimize the distance between two assignments defined by

\begin{equation}
	d(\alpha, \alpha') := \#\{ (n,p) \in \mathbf{N}\times\mathbf{P} ~|~ n\in \alpha_p \triangle \alpha'_p \}
\end{equation}

where the symmetric difference $\alpha_p \triangle \alpha'_p$ denotes the nodes appearing in one of the assignments but not in both.

\section{Computation of an optimal assignment}

The algorithm that we propose takes as inputs the cluster layout parameters $\mathbf{N}$, $\mathbf{Z}$, $\mathbf{P}$, $(c_n)_{n\in \mathbf{N}}$, $\rho_\mathbf{N}$, $\rho_\mathbf{Z}$, that we defined in the introduction, together with the former assignment $\alpha'$ (if any). The computation of the new optimal assignment $\alpha^*$ is done in three successive steps that will be detailed in the following sections. The first step computes the largest partition size $s^*$ that an assignment can achieve. The second step computes an optimal candidate assignment $\alpha$ that achieves $s^*$ and a heuristic is used in the computation to make it hopefully close to $\alpha'$. The third steps modifies $\alpha$ iteratively to reduces $d(\alpha, \alpha')$ and yields an assignment $\alpha^*$ achieving $s^*$, and minimizing $d(\cdot, \alpha')$ among such assignments.

We will explain in the next section how to represent an assignment $\alpha$ by a flow $f$ on a weighted graph $G$ to enable the use of flow and graph algorithms. The main function of the algorithm can be written as follows.

\paragraph{Algorithm}

\begin{algorithmic}[1]
	\Function{Compute Layout}{$\mathbf{N}$, $\mathbf{Z}$, $\mathbf{P}$, $(c_n)_{n\in \mathbf{N}}$, $\rho_\mathbf{N}$, $\rho_\mathbf{Z}$, $\alpha'$}
	\State $s^* \leftarrow$ \Call{Compute Partition Size}{$\mathbf{N}$, $\mathbf{Z}$, $\mathbf{P}$, $(c_n)_{n\in \mathbf{N}}$, $\rho_\mathbf{N}$, $\rho_\mathbf{Z}$}
	\State $G \leftarrow G(s^*)$
	\State $f \leftarrow$ 	\Call{Compute Candidate Assignment}{$G$, $\alpha'$}
	\State $f^* \leftarrow$ 	\Call{Minimize transfer load}{$G$, $f$, $\alpha'$}
	\State Build $\alpha^*$ from $f^*$
	\State \Return $\alpha^*$
	\EndFunction
\end{algorithmic}

\paragraph{Complexity}
As we will see in the next sections, the worst case complexity of this algorithm is $O(P^2 N^2)$. The minimization of transfer load is the most expensive step, and it can run with a timeout since it is only an optimization step. Without this step (or with a smart timeout), the worst case complexity can be $O((PN)^{3/2}\log C)$ where $C$ is the total storage capacity of the cluster. 

\subsection{Determination of the partition size $s^*$}

We will represent an assignment $\alpha$ as a flow in a specific graph $G$. Remark that such flow must have value $\rho_\mathbf{N}P$. We will not compute the optimal partition size $s^*$ a priori, but we will determine it by dichotomy, as the largest size $s$ such that the maximal flow achievable on $G=G(s)$ has value $\rho_\mathbf{N}P$. We will assume that the capacities are given in a small enough unit (e.g.~megabytes), and we will determine $s^*$ at the precision of the given unit.

Given some candidate size value $s$, we describe the oriented weighted graph $G=(V,E)$ with vertex set $V$ and arc set $E$ (see Figure \ref{fig:flowgraph}).

The set of vertices $V$ contains the source $\mathbf{s}$, the sink $\mathbf{t}$, vertices 
$\mathbf{p^+, p^-}$ for every partition $p$, vertices $\mathbf{x}_{p,z}$ for every partition $p$ and zone $z$, and vertices $\mathbf{n}$ for every node $n$. 

The set of arcs $E$ contains:
\begin{itemize}
	\item ($\mathbf{s}$,$\mathbf{p}^+$, $\rho_\mathbf{Z}$) for every partition $p$;
	\item ($\mathbf{s}$,$\mathbf{p}^-$, $\rho_\mathbf{N}-\rho_\mathbf{Z}$) for every partition $p$;
	\item ($\mathbf{p}^+$,$\mathbf{x}_{p,z}$, 1) for every partition $p$ and zone $z$;
	\item ($\mathbf{p}^-$,$\mathbf{x}_{p,z}$, $\rho_\mathbf{N}-\rho_\mathbf{Z}$) for every partition $p$ and zone $z$;
	\item ($\mathbf{x}_{p,z}$,$\mathbf{n}$, 1) for every partition $p$, zone $z$ and node $n\in z$;
	\item ($\mathbf{n}$, $\mathbf{t}$, $\lfloor c_n/s \rfloor$) for every node $n$.
\end{itemize}

\begin{figure}
	\centering
	\includegraphics[width=\linewidth]{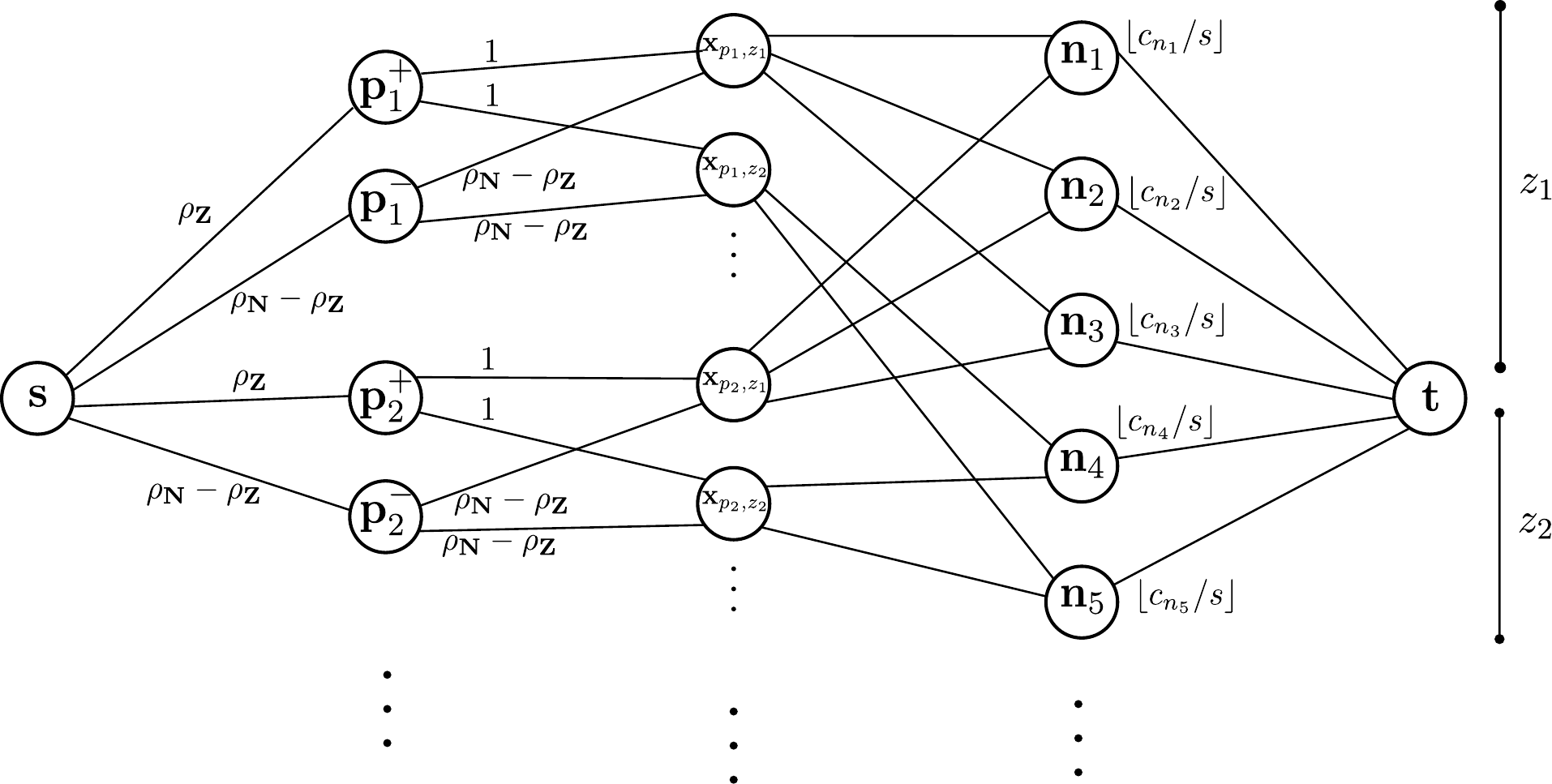}
	\caption{An example of graph $G(s)$. Arcs are oriented from left to right, and unlabeled arcs have capacity 1. In this example, nodes $n_1,n_2,n_3$ belong to zone $z_1$, and nodes $n_4,n_5$ belong to zone $z_2$.}
	\label{fig:flowgraph}
\end{figure}

In the following complexity calculations, we will use the number of vertices and edges of $G$. Remark for now that $\# V = O(PZ)$ and $\# E = O(PN)$.

\begin{proposition}
	An assignment $\alpha$ is realizable with partition size $s$ and replication and scattering factors $(\rho_\mathbf{N},\rho_\mathbf{Z})$ if and only if there exists a maximal flow function $f$ in $G$ with total flow $\rho_\mathbf{N}P$, such that the arcs ($\mathbf{x}_{p,z}$,$\mathbf{n}$, 1) used are exactly those for which $p$ is associated to $n$ in $\alpha$.
\end{proposition}
\begin{proof}
	Given such flow $f$, we can reconstruct a candidate $\alpha$. In $f$, the flow passing through $\mathbf{p^+}$ and $\mathbf{p^-}$ is $\rho_\mathbf{N}$, and since the outgoing capacity of every $\mathbf{x}_{p,z}$ is 1, every partition is associated to $\rho_\mathbf{N}$ distinct nodes. The fraction $\rho_\mathbf{Z}$ of the flow passing through every $\mathbf{p^+}$ must be spread over as many distinct zones as every arc outgoing from $\mathbf{p^+}$ has capacity 1. So the reconstructed $\alpha$ verifies the replication and scattering constraints. For every node $n$, the flow between $\mathbf{n}$ and $\mathbf{t}$ corresponds to the number of partitions associated to $n$. By construction of $f$, this does not exceed $\lfloor c_n/s \rfloor$. We assumed that the partition size is $s$, hence this association does not exceed the storage capacity of the nodes.
	
	In the other direction, given an assignment $\alpha$, one can similarly check that the facts that $\alpha$ respects the replication and scattering constraints, and the storage capacities of the nodes, are necessary condition to construct a maximal flow function $f$.
\end{proof}

\paragraph{Implementation remark.} In the flow algorithm, while exploring the graph, we explore the neighbours of every vertex in a random order to heuristically spread the associations between nodes and partitions.

\paragraph{Algorithm}
With this result mind, we can describe the first step of our algorithm. All divisions are supposed to be integer divisions.
\begin{algorithmic}[1]
	\Function{Compute Partition Size}{$\mathbf{N}$, $\mathbf{Z}$, $\mathbf{P}$, $(c_n)_{n\in \mathbf{N}}$, $\rho_\mathbf{N}$, $\rho_\mathbf{Z}$}
	
	\State Build the graph $G=G(s=1)$
	\State $ f \leftarrow$ \Call{Maximal flow}{$G$}
	\If{$f.\mathrm{total flow} < \rho_\mathbf{N}P$}
	
	\State \Return Error: capacities too small or constraints too strong.
	\EndIf
	
	\State $s^- \leftarrow 1$
	\State $s^+ \leftarrow 1+\frac{1}{\rho_\mathbf{N}}\sum_{n \in \mathbf{N}} c_n$
	
	\While{$s^-+1 < s^+$}
	\State Build the graph $G=G(s=(s^-+s^+)/2)$
	\State $ f \leftarrow$ \Call{Maximal flow}{$G$}
	\If{$f.\mathrm{total flow} < \rho_\mathbf{N}P$}
	\State $s^+ \leftarrow (s^- + s^+)/2$
	\Else
	\State $s^- \leftarrow (s^- + s^+)/2$
	\EndIf
	\EndWhile
	
	\State \Return $s^-$
	\EndFunction
\end{algorithmic}

\paragraph{Complexity}
To compute the maximal flow, we use Dinic's algorithm \cite{Dinitz}. Its complexity on general graphs is $O(\#V^2 \#E)$, but on graphs with edge capacity bounded by a constant, it turns out to be $O(\#E^{3/2})$. The graph $G$ does not fall in this case since the capacities of the arcs incoming to $\mathbf{t}$ are far from bounded. However, the proof of this complexity function works readily for graphs where we only ask the edges \emph{not} incoming to the sink $\mathbf{t}$ to have their capacities bounded by a constant. One can find the proof of this claim in \cite[Section 2]{even1975network}.
The dichotomy adds a logarithmic factor $\log (C)$ where $C=\sum_{n \in \mathbf{N}} c_n$ is the total capacity of the cluster. The total complexity of this first function is hence 
$O(\#E^{3/2}\log C ) = O\big((PN)^{3/2} \log C\big)$.

\paragraph{Metrics}
We can display the discrepancy between the computed $s^*$ and the best size we could have hoped for the given total capacity, that is $C/\rho_\mathbf{N}$.

\subsection{Computation of a candidate assignment}

Now that we have the optimal partition size $s^*$, to compute a candidate assignment it would be enough to compute a maximal flow function $f$ on $G(s^*)$. This is what we do if there is no former assignment $\alpha'$.

If there is some $\alpha'$, we add a step that will heuristically help to obtain a candidate $\alpha$ closer to $\alpha'$. We fist compute a flow function $\tilde{f}$ that uses only the partition-to-node associations appearing in $\alpha'$. Most likely, $\tilde{f}$ will not be a maximal flow of $G(s^*)$. In Dinic's algorithm, we can start from a non maximal flow function and then discover improving paths. This is what we do by starting from $\tilde{f}$. The hope\footnote{This is only a hope, because one can find examples where the construction of $f$ from $\tilde{f}$ produces an assignment $\alpha$ that is not as close as possible to $\alpha'$.} is that the final flow function $f$ will tend to keep the associations appearing in $\tilde{f}$.

More formally, we construct the graph $G_{|\alpha'}$ from $G$ by removing all the arcs $(\mathbf{x}_{p,z},\mathbf{n}, 1)$ where $p$ is not associated to $n$ in $\alpha'$. We compute a maximal flow function $\tilde{f}$ in $G_{|\alpha'}$. The flow $\tilde{f}$ is also a valid (most likely non maximal) flow function on $G$. We compute a maximal flow function $f$ on $G$ by starting Dinic's algorithm with $\tilde{f}$.

\paragraph{Algorithm}
\begin{algorithmic}[1]
	\Function{Compute Candidate Assignment}{$G$, $\alpha'$}
	\State Build the graph $G_{|\alpha'}$
	\State $ \tilde{f} \leftarrow$ \Call{Maximal flow}{$G_{|\alpha'}$}
	\State $ f \leftarrow$ \Call{Maximal flow from flow}{$G$, $\tilde{f}$}
	\State \Return $f$
	\EndFunction
\end{algorithmic}

\paragraph{Remark}
The function ``Maximal flow'' can be just seen as the function ``Maximal flow from flow'' called with the zero flow function as starting flow.

\paragraph{Complexity}
With the considerations of the last section, we have the complexity of Dinic's algorithm $O(\#E^{3/2}) = O((PN)^{3/2})$.

\paragraph{Metrics}

We can display the flow value of $\tilde{f}$, which is an upper bound of the distance between $\alpha$ and $\alpha'$, although this information might not be very relevant to end users.

\subsection{Minimization of the transfer load}

Now that we have a candidate flow function $f$, we want to modify it to make its corresponding assignment $\alpha$ as close as possible to $\alpha'$. Denote by $f'$ the maximal flow corresponding to $\alpha'$, and let $d(f, \alpha')=d(f, f'):=d(\alpha,\alpha')$\footnote{It is the number of arcs of type $(\mathbf{x}_{p,z},\mathbf{n})$ saturated in one flow and not in the other.}. 
We want to build a sequence $f=f_0, f_1, f_2 \dots$ of maximal flows such that $d(f_i, \alpha')$ decreases as $i$ increases. The distance being a non-negative integer, this sequence of flow functions must be finite. We now explain how to find some improving $f_{i+1}$ from $f_i$.

For any maximal flow $f$ in $G$, we define the oriented weighted graph $G_f=(V, E_f)$ as follows. The vertices of $G_f$ are the same as the vertices of $G$. $E_f$ contains the arc $(v_1,v_2, w)$ between vertices $v_1,v_2\in V$ with weight $w$ if and only if the arc $(v_1,v_2)$ is not saturated in $f$ (i.e. $c(v_1,v_2)-f(v_1,v_2) \ge 1$, we also consider reversed arcs). The weight $w$ is: 
\begin{itemize}
	\item  $-1$  if $(v_1,v_2)$ is of type $(\mathbf{x}_{p,z},\mathbf{n})$ or $(\mathbf{n},\mathbf{x}_{p,z})$ and is saturated in only one of the two flows $f,f'$;
	\item $+1$  if $(v_1,v_2)$ is of type $(\mathbf{x}_{p,z},\mathbf{n})$ or $(\mathbf{n},\mathbf{x}_{p,z})$ and is saturated in either both or none of the two flows $f,f'$;
	\item $0$ otherwise.
\end{itemize}

If $\gamma$ is a simple cycle of arcs in $G_f$, we define its weight $w(\gamma)$ as the sum of the weights of its arcs. We can add $+1$ to the value of $f$ on the arcs of $\gamma$, and by construction of $G_f$ and the fact that $\gamma$ is a cycle, the function that we get is still a valid flow function on $G$, it is maximal as it has the same flow value as $f$. We denote this new function $f+\gamma$.

\begin{proposition}
	Given a maximal flow $f$ and a simple cycle $\gamma$ in $G_f$, we have $d(f+\gamma, f') - d(f,f') = w(\gamma)$.
\end{proposition}
\begin{proof}
	Let $X$ be the set of arcs of type $(\mathbf{x}_{p,z},\mathbf{n})$. Then we can express $d(f,f')$ as
	\begin{align*}
		d(f,f') & = \#\{e\in X ~|~ f(e)\neq f'(e)\}
		= \sum_{e\in X} 1_{f(e)\neq f'(e)} \\
		& = \frac{1}{2}\big( \#X + \sum_{e\in X} 1_{f(e)\neq f'(e)} - 1_{f(e)= f'(e)} \big).
	\end{align*}
	We can express the cycle weight as
	\begin{align*}
		w(\gamma) &  = \sum_{e\in X, e\in \gamma} - 1_{f(e)\neq f'(e)} + 1_{f(e)= f'(e)}.
	\end{align*}
	Remark that since we passed one unit of flow in $\gamma$ to construct $f+\gamma$, we have for any $e\in X$, $f(e)=f'(e)$ if and only if $(f+\gamma)(e) \neq f'(e)$.
	Hence
	\begin{align*}
		w(\gamma) &  = \frac{1}{2}(w(\gamma) + w(\gamma)) \\
		&= \frac{1}{2} \Big( 
		\sum_{e\in X, e\in \gamma} - 1_{f(e)\neq f'(e)} +  1_{f(e)= f'(e)} \\
		& \qquad +
		\sum_{e\in X, e\in \gamma} 1_{(f+\gamma)(e)\neq f'(e)} +  1_{(f+\gamma)(e)= f'(e)}
		\Big).
	\end{align*}
	Plugging this in the previous equation, we find that 
	$$d(f,f')+w(\gamma) = d(f+\gamma, f').$$
\end{proof}

This result suggests that given some flow $f_i$, we just need to find a negative cycle $\gamma$ in $G_{f_i}$ to construct $f_{i+1}$ as $f_i+\gamma$. The following proposition ensures that this greedy strategy reaches an optimal flow. 

\begin{proposition}
	For any maximal flow $f$, $G_f$ contains a negative cycle if and only if there exists a maximal flow $f^*$ in $G$ such that $d(f^*, f') < d(f, f')$.
\end{proposition}
\begin{proof}
	Suppose that there is such flow $f^*$. Define the oriented multigraph $M_{f,f^*}=(V,E_M)$ with the same vertex set $V$ as in $G$, and for every $v_1,v_2 \in V$, $E_M$ contains $(f^*(v_1,v_2) - f(v_1,v_2))_+$  copies of the arc $(v_1,v_2)$. For every vertex $v$, its total degree (meaning its outer degree minus its inner degree) is equal to 
	\begin{align*}
		\deg v & = \sum_{u\in V} (f^*(v,u) - f(v,u))_+ - \sum_{u\in V} (f^*(u,v) - f(u,v))_+  \\
		& = \sum_{u\in V} f^*(v,u) - f(v,u) = \sum_{u\in V} f^*(v,u) - \sum_{u\in V}  f(v,u).
	\end{align*}
	The last two sums are zero for any inner vertex since $f,f^*$ are flows, and they are equal on the source and sink since the two flows are both maximal and have hence the same value. Thus, $\deg v = 0$ for every vertex $v$.
	
	This implies that the multigraph $M_{f,f^*}$ is the union of disjoint simple cycles. $f$ can be transformed into $f^*$ by pushing a mass 1 along all these cycles in any order. Since $d(f^*, f')<d(f,f')$, there must exists one of these simple cycles $\gamma$ with $d(f+\gamma, f') < d(f, f')$. Finally, since we can push a mass in $f$ along $\gamma$, it must appear in $G_f$. Hence $\gamma$ is a cycle of $G_f$ with negative weight.
\end{proof}

In the next section we describe the corresponding algorithm. Instead of discovering only one cycle per iteration, we are allowed to discover a set $\Gamma$ of disjoint negative cycles. 

\paragraph{Algorithm}
\begin{algorithmic}[1]
	\Function{Minimize transfer load}{$G$, $f$, $\alpha'$}
	\State Build the graph $G_f$
	\State $\Gamma \leftarrow$ \Call{Detect Negative Cycles}{$G_f$}
	\While{$\Gamma \neq \emptyset$}
	\ForAll{$\gamma \in \Gamma$}
	\State $f \leftarrow f+\gamma$
	\EndFor
	\State Update $G_f$
	\State $\Gamma \leftarrow$ \Call{Detect Negative Cycles}{$G_f$}
	\EndWhile
	\State \Return $f$
	\EndFunction
\end{algorithmic}

\paragraph{Complexity}
The distance $d(f,f')$ is bounded by the maximal number of differences in the associated assignment. If these assignment are totally disjoint, this distance is $2\rho_N P$. At every iteration of the While loop, the distance decreases, so there is at most $O(\rho_N P) = O(P)$ iterations.

The detection of negative cycles is done with the Bellman-Ford algorithm, whose complexity should normally be $O(\#E\#V)$. In our case, it amounts to $O(P^2ZN)$. Multiplied by the complexity of the outer loop, it amounts to $O(P^3ZN)$ which is a lot when the number of partitions and nodes starts to be large. To avoid that, we adapt the Bellman-Ford algorithm.

The Bellman-Ford algorithm runs $\#V$ iterations of an outer loop, and an inner loop over $E$. The idea is to compute the shortest paths from a source vertex $v$ to all other vertices. After $k$ iterations of the outer loop, the algorithm has computed all shortest path of length at most $k$. All simple paths have length at most $\#V-1$, so if there is an update in the last iteration of the loop, it means that there is a negative cycle in the graph. The observation that will enable us to improve the complexity is the following:

\begin{proposition}
	In the graph $G_f$ (and $G$), all simple paths have a length at most $4N$.
\end{proposition}
\begin{proof}
	Since $f$ is a maximal flow, there is no outgoing edge from $\mathbf{s}$ in $G_f$. One can thus check than any simple path of length 4 must contain at least two node of type $\mathbf{n}$. Hence on a path, at most 4 arcs separate two successive nodes of type $\mathbf{n}$. 
\end{proof}

Thus, in the absence of negative cycles, shortest paths in $G_f$ have length at most $4N$. So we can do only $4N+1$ iterations of the outer loop in the Bellman-Ford algorithm. This makes the complexity of the detection of one set of cycle to be $O(N\#E) = O(N^2 P)$.

With this improvement, the complexity of the whole algorithm is, in the worst case, $O(N^2P^2)$. However, since we detect several cycles at once and we start with a flow that might be close to the previous one, the number of iterations of the outer loop might be smaller in practice.

\paragraph{Metrics}
We can display the node and zone utilization ratio, by dividing the flow passing through them divided by their outgoing capacity. In particular, we can pinpoint saturated nodes and zones (i.e. used at their full potential).

We can display the distance to the previous assignment, and the number of partition transfers.

\section{Related work}

In previous versions of Garage, we iterated through many algorithms to build an assignment of partitions to nodes, always with unsatisfactory results.
These previous attempts, all based on existing work, are described in this section.

\paragraph{Basic consistent hashing with zone awareness}
In this algorithm, we use the simple consistent hashing ring described in Dynamo~\cite{decandia2007dynamo}.
We slightly adapt it to support nodes in different zones and the requirement to spread replicas over as many zones as possible:
when looking up the nodes associated to a data block, we walk the ring starting from the position corresponding to its hash, but we skip
nodes that are in a zone from which we have already selected a node (except if there are no more distinct zones to take nodes from).
This method had the disadvantage of giving a very unbalanced distribution of data between nodes. For example, suppose that there are many consecutive
nodes on the ring that are in zones 1 and 2, followed by one node in zone 3. Then that node will store a copy of all data blocks whose hashes are in the
interval before it that contains only nodes of zone 1 and 2.

\paragraph{Arbitrary ring positions vs.~fixed partition boundaries}
As already discussed in the Dynamo paper~\cite{decandia2007dynamo} (see the three different strategies presented in Figure~7),
using the hashes of node identifiers as positions on the consistent hashing ring
makes the intervals between these positions of wildly varying sizes, worsening the imbalance of storage affected to all nodes.
To resolve this issue, we very rapidly switched to dividing the consistent hashing ring into equally sized parts (what we call partitions),
as shown in Dynamo's strategies 2 and 3.
To ensure that all nodes handle a number of partitions strictly proportional to their capacity,
we tried using the MagLev algorithm~\cite{eisenbud2016maglev} to assign partitions to nodes.
However, just doing this does not solve the zone awareness issue; continuing to use the simple ring walking where nodes are skipped still produces
a very imbalanced distribution.

\paragraph{Multi-zone aware MagLev}
Our next try was to improve the MagLev algorithm to be multi-zone aware.
Now, instead of assigning a single node to each ring position (each partition) and walking the ring to find three nodes starting at a given key's hash,
we directly assign a set of three nodes to each partition and completely abandon ring walking.
The first node of the three is computed for all partitions by using the standard MagLev algorithm.
Then, the next two are computed using a variant of MagLev that skips assigning nodes to partitions when they are in zones of nodes already selected
for that partition (unless there are no more distinct zones available), selecting other nodes instead.
This way, we ensure that the three nodes assigned to each partition are in as many distinct zones as possible.
This method provided perfectly equitable distribution of data among nodes, however when layout changes occurred, the entire assignment was recomputed
without taking into account the previous one, and thus there was no way to ensure that a minimal amount of data was displaced from one node to another.

\paragraph{Stateful assignment algorithms}
In all of the previous iterations, we were limiting ourselves to algorithms that were stateless:
the assignment had to be computed in a deterministic way from only the list of node identifiers and their zone and capacity information,
using hash functions to provide pseudo-randomness.
To be able to minimize the transfer load on layout changes, we had to switch to a stateful method where the entire assignment is computed offline and then propagated to all cluster nodes.
It can now be computed using any arbitrary optimization algorithm that can take as an input the previous assignment to minimize transfer load.
This method was introduced in Garage version 0.5 with a simple greedy optimization algorithm that was not optimal, which was in use until version 0.8.
The final, optimal assignment algorithm is the one we presented in this paper, which will be included in Garage version 0.9 and forward.

\section*{Acknowledgements}
This project has received funding from the European Union's Horizon 2021 research and innovation programme within the framework of the NGI-POINTER Project funded under grant agreement N° 871528.

\bibliography{geodistrib} 
\bibliographystyle{ieeetr}

\end{document}